\documentclass{llncs}
\usepackage{pta}

\title{Probabilistic Thread Algebra}
\author{J.A. Bergstra \and C.A. Middelburg}
\institute{Informatics Institute, Faculty of Science, University of
           Amsterdam, \\
           Science Park~904, 1098~XH Amsterdam, the Netherlands \\
           \email{J.A.Bergstra@uva.nl,C.A.Middelburg@uva.nl}}

\begin{document}
\maketitle

\begin{abstract}
We add probabilistic features to basic thread algebra and its extensions 
with thread-service interaction and strategic interleaving.
Here, threads represent the behaviours produced by instruction sequences 
under execution and services represent the behaviours exhibited by the 
components of execution environments of instruction sequences.
In a paper concerned with probabilistic instruction sequences, we 
proposed several kinds of probabilistic instructions and gave an 
informal explanation for each of them.
The probabilistic features added to the extension of basic thread 
algebra with thread-service interaction make it possible to give a 
formal explanation in terms of non-probabilistic instructions and 
probabilistic services.
The probabilistic features added to the extensions of basic thread 
algebra with strategic interleaving make it possible to cover strategies
corresponding to probabilistic scheduling algorithms.
\begin{classcode}
D.3.3, D.4.1, F.1.1, F.1.2.
\end{classcode}
\end{abstract}

\section{Introduction}
\label{sect-intro}

In~\cite{BL02a}, an approach to the semantics of programming languages 
was presented which is based on the perspective that a program is in 
essence an instruction sequence.
The groundwork for the approach is formed by \PGA\ (ProGram Algebra), 
an algebraic theory of single-pass instruction sequences, and \BTA\ 
(Basic Thread Algebra), an algebraic theory of mathematical objects that 
represent the behaviours produced by instruction sequences under 
execution (for a comprehensive introduction to these algebraic theories, 
see~\cite{BM12b}).
To increase the applicability of the approach, \BTA\ was extended with
thread-service interaction in~\cite{BP02a}.
In the setting of \BTA\ and its extension with thread-service 
interaction, threads are mathematical objects that represent the 
behaviours produced by instruction sequences under execution and 
services are mathematical objects that represent the behaviours 
exhibited by components of execution environments of instruction 
sequences.

As a continuation of the work presented in~\cite{BL02a,BP02a},
(a)~the notion of an instruction sequence was subjected to systematic
and precise analysis using the groundwork laid earlier,
(b)~various issues, including issues relating to computability and 
complexity of computational problems, efficiency of algorithms, and
verification of programs, were rigorously investigated thinking in terms 
of instruction sequences (for a comprehensive survey of a large part of 
the work referred to under~(a) and~(b), see~\cite{BM12b}), and 
(c)~the form of interleaving concurrency that is relevant to the 
behaviours of multi-threaded programs under execution, called strategic
interleaving in the setting of \BTA, was rigorously investigated by 
means of extensions of \BTA\ (see e.g.~\cite{BM04c,BM06a,BM07a}). 

In the course of the work referred to above under~(b), we ran into the 
problem that \BTA\ and its extension with thread-service interaction 
do not allow issues relating to probabilistic computation to be 
investigated thinking in terms of instruction sequences.  
In the course of the work referred to above under~(c), we ran into the 
problem that \BTA\ also does not allow probabilistic strategic 
interleaving to be investigated by means of extensions of \BTA.
This paper concerns the addition of features to \BTA\ and its extensions 
with thread-service interaction and strategic interleaving that will 
take away these limitations.

We consider it important to take probabilistic computation into account
in future investigations.
The primary reasons for this are the following:
(a)~the existence of probabilistic algorithms that are highly efficient, 
possibly at the cost of a probability of correctness less than one 
(e.g.\ primality testing, see~\cite{Rab80a});
(b)~the existence of probabilistic algorithms for which no deterministic 
counterparts exist (e.g.\ symmetry breaking, see~\cite{IR90a});
(c)~the gradually created evidence for the hypothesis that it is 
relevant for a diversity of issues in computer science and engineering 
to think in terms of instruction sequences.
This constitutes the basis of our motivation for the work presented in 
this paper. 

In~\cite{BM09f}, we gave an enumeration of kinds of probabilistic 
instructions that were chosen on the basis of direct intuitions and 
therefore not necessarily the best kinds in any sense.
We only gave an informal explanation for each of the enumerated kinds
because we considered it premature at the time to add probabilistic 
features to \BTA\ that would make it possible to give a formal 
explanation.
We were doubtful whether the ad hoc addition of features to \BTA\ was 
the right way to go.

Later, we have found that the ramification of semantic options with 
the addition of probabilistic features to \BTA\ is well surveyable
because of 
(a)~the limitation of the scope to behaviours produced by instruction
sequences under execution and
(b)~the semantic constraints brought about by the informal explanations 
of the kinds of probabilistic instructions enumerated in~\cite{BM09f} 
and the desired elimination property of all but one kind.
In the case of a general process algebra, such as \ACP~\cite{BW90}, 
CCS~\cite{Mil89} or CSP~\cite{Hoa85}, the ramification becomes much more 
complex, particularly because a limitation of the scope to behaviours of 
a special kind is lacking.
In this paper, we add probabilistic features to \BTA\ and an extension 
of \BTA\ with thread-service interaction.

The probabilistic features added to the extension of \BTA\ with 
thread-service interaction make it possible to give a formal explanation 
for each of the kinds of probabilistic instructions enumerated 
in~\cite{BM09f} in terms of non-probabilistic instructions and 
probabilistic services.
To demonstrate this, we add the kind of probabilistic instructions that 
cannot be eliminated to \PGLB\ (ProGramming Language B), a program 
notation rooted in \PGA\ and close to existing assembly languages, and 
give a formal definition of the behaviours produced by the instruction 
sequences from the resulting program notation.
We opted for \PGLB\ because in the past it has proved itself suitable 
for the investigation of various issues.
The added kind of probabilistic instructions allow probabilistic choices 
to be made during the execution of instruction sequences. 

In~\cite{BM04c} and subsequent papers, we extended \BTA\ with kinds of 
interleaving where interleaving takes place according to some 
deterministic interleaving strategy.
Interleaving strategies are abstractions of scheduling algorithms.
Interleaving according to an interleaving strategy differs from 
arbitrary interleaving, but it is what really happens in the case of 
multi-threading as found in programming languages such as 
Java~\cite{GJSB00a} and C\#~\cite{HWG03a}.
The extension of \BTA\ with a probabilistic feature does not only allow
of probabilistic services, but also allows of probabilistic interleaving
strategies.
In this paper, we also generalize the extensions of \BTA\ with specific 
kinds of deterministic strategic interleaving to an extension for a 
large class of kinds of deterministic and probabilistic strategic 
interleaving.
Thus, strategies corresponding to probabilistic scheduling algorithms
such as the lottery scheduling algorithm~\cite{WW94a} are covered.

The main results of this paper are probabilistic versions of \BTA\ and
its extensions with thread-service interaction and strategic 
interleaving which pave the way for  
(a)~investigation of issues related to probabilistic computation 
thinking in terms of instruction sequences and 
(b)~investigation of probabilistic interleaving strategies.

In this paper, we take functions whose range is the carrier of a signed 
cancellation meadow as probability measures.
In~\cite{BT07a}, meadows are proposed as alternatives for fields with a
purely equational axiomatization.
A meadow is a commutative ring with a multiplicative identity element
and a total multiplicative inverse operation satisfying two equations
which imply that the multiplicative inverse of zero is zero.
A cancellation meadow is a field whose multiplicative inverse operation 
is made total by imposing that the multiplicative inverse of zero is 
zero, and a signed cancellation meadow is a cancellation meadow expanded 
with a signum operation.
In~\cite{BP13a}, Kolmogorov's probability axioms for finitely additive 
probability spaces are rephrased for the case where probability measures 
are functions whose range is the carrier of a signed cancellation 
meadow.

This paper is organized as follows.
First, we review signed cancellation meadows 
(Section~\ref{sect-meadows}). 
Next, we add probabilistic features to \BTA\ and an extension of \BTA\ 
with thread-service interaction (Sections~\ref{sect-prBTA}
and~\ref{sect-TSI}).
Then, we add a kind of probabilistic instructions to \PGLB\
(Section~\ref{sect-prPGLB}).
Following this, we add probabilistic features to the extensions of \BTA\ 
with strategic interleaving 
(Section~\ref{sect-strategic-interleaving}).
Finally, we make some concluding remarks (Section~\ref{sect-concl}).

It should be mentioned that \BTA\ is introduced in~\cite{BL02a} under 
the name BPPA (Basic Polarized Process Algebra) and services are called
state machines in~\cite{BP02a}.

\section{Signed Cancellation Meadows}
\label{sect-meadows}

We will take functions whose range is the carrier of a signed 
cancellation meadow as probability measures.
Therefore, we review signed cancellation meadows in this section.

In~\cite{BT07a}, meadows are proposed as alternatives for fields with a
purely equational axiomatization.
A meadow is a commutative ring with a multiplicative identity element
and a total multiplicative inverse operation satisfying two equations
which imply that the multiplicative inverse of zero is zero.
Thus, all meadows are total algebras and the class of all meadows is a 
variety.
At the basis of meadows lies the decision to make the multiplicative 
inverse operation total by imposing that the multiplicative inverse of 
zero is zero.
All fields in which the multiplicative inverse of zero is zero, called
zero-totalized fields, are meadows, but not conversely.

A cancellation meadow is a meadow that satisfies the 
\emph{cancellation axiom} 
$ x \neq 0 \Land x \mul y = x \mul z \Limpl y = z$.
The zero-totalized fields are exactly the cancellation meadows that
satisfy in addition the \emph{separation axiom} $0 \neq 1$.
A paradigmatic example of cancellation meadows is the field of rational
numbers with the multiplicative inverse operation made total by imposing
that the multiplicative inverse of zero is zero (see e.g.~\cite{BT07a}).
An example of a meadow that is not a zero-totalized field is the initial
algebra of the equational axiomatization of meadows 
(see e.g.~\cite{BHT09a}).

A signed cancellation meadow is a cancellation meadow expanded with a
signum operation.
The usefulness of the signum operation lies in the fact that the 
predicates $<$ and $\leq$ can be defined using this operation 
(see below).

The signature of signed cancellation meadows consists of the following
constants and operators:
the constants $0$ and $1$,
the binary \emph{addition} operator ${} +$ {},
the binary \emph{multiplication} operator ${} \mul {}$,
the unary \emph{additive inverse} operator $- {}$,
the unary \emph{multiplicative inverse} operator ${}\minv$, and
the unary \emph{signum} operator $\sign$.

Terms are built as usual.
We use infix notation for the binary operators ${} + {}$ and
${} \mul {}$, prefix notation for the unary operator $- {}$, and postfix
notation for the unary operator ${}\minv$.
We use the usual precedence convention to reduce the need for
parentheses.
We introduce subtraction and division as abbreviations:
$t - t'$ abbreviates $t + (-t')$ and
$t / t'$ abbreviates $t \mul ({t'}\minv)$.

The constants and operators from the signature of signed cancellation
meadows are adopted from rational arithmetic, which gives an appropriate
intuition about these constants and operators.

Signed cancellation meadows are axiomatized by the equations in 
Tables~\ref{eqns-meadow} and~\ref{eqns-signum} and the above-mentioned
cancellation axiom.
\begin{table}[!t]
\caption
{Axioms of a meadow}
\label{eqns-meadow}
\begin{eqntbl}
\begin{eqncol}
(x + y) + z = x + (y + z)                                             \\
x + y = y + x                                                         \\
x + 0 = x                                                             \\
x + (-x) = 0
\end{eqncol}
\qquad\quad
\begin{eqncol}
(x \mul y) \mul z = x \mul (y \mul z)                                 \\
x \mul y = y \mul x                                                   \\
x \mul 1 = x                                                          \\
x \mul (y + z) = x \mul y + x \mul z
\end{eqncol}
\qquad\quad
\begin{eqncol}
(x\minv)\minv = x                                                     \\
x \mul (x \mul x\minv) = x                           
\end{eqncol}
\end{eqntbl}
\end{table}
\begin{table}[!t]
\caption{Additional axioms for the signum operator}
\label{eqns-signum}
\begin{eqntbl}
\begin{eqncol}
\sign(x / x) = x / x                                                  \\
\sign(1 - x / x) = 1 - x / x                                          \\
\sign(-1) = -1
\end{eqncol}
\qquad\quad
\begin{eqncol}
\sign(x\minv) = \sign(x)                                              \\
\sign(x \mul y) = \sign(x) \mul \sign(y)                              \\
(1 - \frac{\sign(x) - \sign(y)}{\sign(x) - \sign(y)}) \mul
(\sign(x + y) - \sign(x)) = 0
\end{eqncol}
\end{eqntbl}
\end{table}
The axioms for the signum operator stem from~\cite{BBP13a}.

The predicates $<$ and $\leq$ are defined in signed cancellation meadows 
as follows: 
$x < y \Liff \sign(y - x) = 1$ and
$x \leq y \Liff \sign(\sign(y - x) + 1) = 1$.
Because $\sign(\sign(y - x) + 1) \neq -1$, we have
$0 \leq x \leq 1 \Liff 
 \sign(\sign(x) + 1) \mul \sign(\sign(1 - x) + 1) = 1$.
We will use this equivalence below to describe the set of probabilities.

\section{Probabilistic Basic Thread Algebra}
\label{sect-prBTA}

In this section, we introduce \prBTA\ 
(probabilistic Basic Thread Algebra), a probabilistic version of \BTA.
The objects considered in \BTA\ are called threads.
In \BTA, a thread represents a behaviour which consists of performing 
actions in a deterministic sequential fashion.
Upon each action performed, a reply from an execution environment 
determines how the thread proceeds.
The possible replies are the values $\True$ and $\False$.
In \prBTA, a thread represents a behaviour which consists of performing 
actions in a probabilistic sequential fashion.
That is, performing actions may alternate with making internal choices 
according to discrete probability distributions.

In the sequel, it is assumed that a fixed but arbitrary signed 
cancellation meadow $\fM$ has been given.
We denote the carrier of $\fM$ by $\fM$ as well, and we denote the 
interpretations of the constants and operators in $\fM$ by the constants 
and operators themselves.
We write $\Prob$ for the set
$\set{\pi \in \fM \where
 \sign(\sign(\pi) + 1) \mul \sign(\sign(1 - \pi) + 1) = 1}$ 
of probabilities.

In \prBTA, it is moreover assumed that a fixed but arbitrary set $\BAct$ 
of \emph{basic actions}, with $\Tau \not\in \BAct$, has been given.
In addition, there is the special action $\Tau$.
Performing $\Tau$, which is considered performing an internal action,
will always lead to the reply $\True$.
We write $\BActTau$ for $\BAct \union \set{\Tau}$ and refer to the 
members of $\BActTau$ as basic actions.

The algebraic theory \prBTA\ has one sort: the sort $\Thr$ of 
\emph{threads}. 
We make this sort explicit to anticipate the need for many-sortedness
later on.
To build terms of sort $\Thr$, \prBTA\ has the following constants and 
operators:
\begin{itemize}
\item
the \emph{inaction} constant $\const{\DeadEnd}{\Thr}$;%
\footnote
{In earlier work, the inaction constant is sometimes called the deadlock
 constant.}
\item
the \emph{termination} constant $\const{\Stop}{\Thr}$;
\item
for each $a \in \BActTau$, the binary \emph{postconditional composition} 
operator $\funct{\pcc{\ph}{a}{\ph}}{\linebreak[2]\Thr \x \Thr}{\Thr}$;
\item
for each $\pi \in \Prob$, the binary \emph{probabilistic composition} 
operator $\funct{\prc{\ph}{\pi}{\ph}}{\Thr \x \Thr}{\Thr}$.
\end{itemize}
Terms of sort $\Thr$ are built as usual in the one-sorted case. 
We assume that there are infinitely many variables of sort $\Thr$, 
including $x,y,z$.
We use infix notation for postconditional composition and probabilistic
composition. 
We introduce \emph{basic action prefixing} as an abbreviation: 
$a \bapf t$, where $t$ is a \prBTA\ term, abbreviates 
$\pcc{t}{a}{t}$.
We identify expressions of the form $a \bapf t$ with the \prBTA\
terms they stand for.

The thread denoted by a closed term of the form $\pcc{t}{a}{t'}$
will first perform $a$, and then proceed as the thread denoted by
$t$ if the reply from the execution environment is $\True$ and proceed
as the thread denoted by $t'$ if the reply from the execution
environment is $\False$. 
The thread denoted by a closed term of the form $\prc{t}{\pi}{t'}$
will behave like the thread denoted by $t$ with probability $\pi$ and 
like the thread denoted by $t'$ with probability $1 - \pi$.
The thread denoted by $\Stop$ will do no more than terminate and 
the thread denoted by $\DeadEnd$ will become inactive.
A thread becomes inactive if no more basic actions are performed, but
it does not terminate.

The inaction constant, the termination constant and the postconditional 
composition operators are adopted from \BTA.
Counterparts of the probabilistic composition operators are found in 
most probabilistic process algebras that offer probabilistic choices of 
the generative variety (see e.g.~\cite{BBS95a}).

The axioms of \prBTA\ are given in Table~\ref{axioms-prBTA}.%
\begin{table}[!t]
\caption{Axioms of \prBTA}
\label{axioms-prBTA}
\begin{eqntbl}
\begin{axcol}
\pcc{x}{\Tau}{y} = \pcc{x}{\Tau}{x}          & \axiom{\phantom{pr}T1}
\eqnsep
\prc{x}{\pi}{y} = \prc{y}{1{-}\pi}{x}                  & \axiom{prA1} \\
\prc{x}{\pi}{(\prc{y}{\rho}{z})} = 
\prc{(\prc{x}{\frac{\pi}{\pi{+}\rho{-}\pi{\mul}\rho}}{y})}
    {\pi{+}\rho{-}\pi{\mul}\rho}{z}                    & \axiom{prA2} \\
\prc{x}{\pi}{x} = x                                    & \axiom{prA3} \\
\prc{x}{1}{y} = x                                      & \axiom{prA4} 
\end{axcol}
\end{eqntbl}
\end{table}
In this table, $\pi$ and $\rho$ stand for arbitrary probabilities from 
$\Prob$.
Axiom T1 reflects that performing $\Tau$ will always lead to the reply
$\True$ and axioms prA1--prA4 express that probabilistic composition  
provides probabilistic choices of the generative variety 
(see~\cite{GSS95a}).
From prA1 and prA4, we can derive both
$\prc{x}{0}{(\prc{y}{0}{z})} = z$ and $\prc{(\prc{x}{0}{y})}{0}{z} = z$,
and hence also 
$\prc{x}{0}{(\prc{y}{0}{z})} = \prc{(\prc{x}{0}{y})}{0}{z}$.
This last equation can be immediately derived from prA2 as well because 
in meadows $0 / 0 = 0$.

Axiom T1 is adopted from \BTA.
Counterparts of axioms prA1--prA3 are found in most probabilistic 
process algebras that offer probabilistic choices of the generative 
variety (see e.g.~\cite{BBS95a}).
However, in the process algebras concerned the probabilities $0$ and $1$ 
are excluded in probabilistic choices to prevent division by zero.
Owing to this exclusion, axiom prA4 is lacking in these process algebras.

Each closed \prBTA\ term denotes a finite thread, i.e.\ a thread with a
finite upper bound to the number of basic actions that it can perform.
Infinite threads, i.e.\ threads without a finite upper bound to the
number of basic actions that it can perform, can be described by guarded 
recursion.
A \emph{guarded recursive specification} over \prBTA\ is a set of 
recursion equations $E = \set{X = t_X \where X \in V}$, where $V$ is a 
set of variables of sort $\Thr$ and each $t_X$ is a \prBTA\ term in which 
only variables from $V$ occur and each occurrence of a variable in $t_X$ 
is in a subterm of the form $\pcc{t}{a}{t'}$.
We write $\vars(E)$ for the set of all variables that occur on the
left-hand side of an equation in $E$.

We are only interested in models of \prBTA\ in which guarded recursive
specifications have unique solutions.
A model of \prBTA\ in which guarded recursive specifications have unique 
solutions is the projective limit model of \prBTA.
This model is constructed along the same line as the projective limit 
model of \BTA\ presented in~\cite{BM12b}.
It is based on the view that two threads are identical if their 
approximations up to any finite depth are identical. 
The approximation up to depth $n$ of a thread is obtained by cutting it 
off after it has performed $n$ actions if it has not yet terminated or
become inactive.

We confine ourselves to the projective limit model of \prBTA, which has 
an initial model of \prBTA\ as a submodel, for the interpretation of 
\prBTA\ terms. 
An outline of this model is given in Appendix~\ref{app-prBTA}.
In the sequel, we use the term \emph{probabilistic thread} or simply 
\emph{thread} for the elements of the carrier of the model.
Regular threads, i.e.\ finite or infinite threads that can only be in a 
finite number of states, can be defined by means of a finite guarded 
recursive specification.

We extend \prBTA\ with guarded recursion by adding constants for 
solutions of guarded recursive specifications and axioms concerning 
these additional constants.
For each guarded recursive specification $E$ and each $X \in \vars(E)$,
we add a constant standing for the unique solution of $E$ for $X$ to the
constants of \prBTA.
The constant standing for the unique solution of $E$ for $X$ is denoted
by $\rec{X}{E}$.
Moreover, we use the following notation.
Let $t$ be a \prBTA\ term and $E$ be a guarded recursive specification.
Then we write $\rec{t}{E}$ for $t$ with, for all $X \in \vars(E)$, all
occurrences of $X$ in $t$ replaced by $\rec{X}{E}$.
We add the axioms for guarded recursion given in Table~\ref{axioms-rec}
to the axioms of \prBTA.%
\begin{table}[!t]
\caption{Axioms for the guarded recursion constants}
\label{axioms-rec}
\begin{eqntbl}
\begin{saxcol}
\rec{X}{E} = \rec{t_X}{E}  & \mif X \!=\! t_X \in E & \axiom{RDP} \\
E \Limpl X = \rec{X}{E}    & \mif X \in \vars(E)    & \axiom{RSP}
\end{saxcol}
\end{eqntbl}
\end{table}
In this table, $X$, $t_X$ and $E$ stand for an arbitrary variable of 
sort $\Thr$, an arbitrary \prBTA\ term and an arbitrary guarded 
recursive specification, respectively.
Side conditions are added to restrict the variables, terms and guarded
recursive specifications for which $X$, $t_X$ and $E$ stand.

The additional axioms for guarded recursion are known as the recursive
definition principle (RDP) and the recursive specification principle
(RSP).
The equations $\rec{X}{E} = \rec{t_X}{E}$ for a fixed $E$ express that
the constants $\rec{X}{E}$ make up a solution of $E$.
The conditional equations $E \Limpl X = \rec{X}{E}$ express that this
solution is the only one.

In Section~\ref{sect-strategic-interleaving}, we will use the notation 
$\PRC{i=k}{n}{\pi_i}{t_i}$ with $1 \leq k \leq n$ and 
$\sum_{i=k}^{n} \pi_i = 1$ for right-nested probabilistic composition.
The term $\PRC{i=k}{n}{\pi_i}{t_i}$ with $1 \leq k \leq n$ is defined by
induction on $n - k$ as follows:
\begin{ldispl}
\begin{aceqns}
\PRC{i=k}{n}{\pi_i}{t_i} & = & t_k &
\mathrm{if}\; k = n\;,
\\
\PRC{i=k}{n}{\pi_i}{t_i} & = &
 \prc{t_k}{\pi_k}{(\PRC{i=k+1}{n}{\frac{\pi_i}{1-\pi_k}}{t_i})} &
\mathrm{if}\; k < n\;.
\end{aceqns}
\end{ldispl}%
The thread denoted by $\PRC{i=k}{n}{\pi_i}{t_i}$ will behave like the 
thread denoted by $t_k$ with probability $\pi_k$ and \ldots \,and like 
the thread denoted by $t_n$ with proba\-bility~$\pi_n$.

\section{Interaction of Threads with Services}
\label{sect-TSI}

Services are objects that represent the behaviours exhibited by 
components of execution environments of instruction sequences at a high 
level of abstraction.
A service is able to process certain methods.
The processing of a method may involve a change of the service.
At completion of the processing of a method, the service produces a
reply value.
Execution environments are considered to provide a family of 
uniquely-named services.
A thread may interact with the named services from the service family 
provided by an execution environment.
That is, a thread may perform a basic action for the purpose of 
requesting a named service to process a method and to return a reply 
value at completion of the processing of the method.
In this section, we extend \prBTA\ with services, service families, a 
composition operator for service families, an operator that is 
concerned with this kind of interaction, and a general operator for 
abstraction from the internal action $\Tau$.

In \SFA, the algebraic theory of service families introduced below, it 
is assumed that a fixed but arbitrary set $\Meth$ of \emph{methods} has 
been given.
Moreover, the following is assumed with respect to services:
\begin{itemize}
\item
a signature $\Sig{\ServAlg}$ has been given that includes the following
sorts:
\begin{itemize}
\item
the sort $\Serv$ of
\emph{services};
\item
the sort $\BVal$ of \emph{Boolean values};
\end{itemize}
and the following constants and operators:
\begin{itemize}
\item
the
\emph{empty service} constant $\const{\emptyserv}{\Serv}$;
\item
the \emph{reply} constants $\const{\True,\False}{\BVal}$;
\item
for each $m \in \Meth$, the
\emph{derived service} operator $\funct{\derive{m}}{\Serv}{\Serv}$;
\item
for each $m \in \Meth$ and $\pi \in \Prob$, the
\emph{service reply} operator $\funct{\prsreply{m}{\pi}}{\Serv}{\BVal}$;
\end{itemize}
\item
a minimal $\Sig{\ServAlg}$-algebra $\ServAlg$ has been given in which 
the following holds:
\begin{itemize}
\item
$\True \neq \False$; 
\item
$\LAND{m \in \Meth}{}
 (\derive{m}(s) = \emptyserv \Liff 
  \LAND{\pi \in \Prob}{} \prsreply{m}{\pi}(s) = \False)$;
\item
$\LAND{\pi,\rho \in \Prob}{} 
  (\prsreply{m}{\pi}(s) = \True \Land \prsreply{m}{\rho}(s) = \True
    \Limpl \pi = \rho)$.
\end{itemize}
\end{itemize}

The intuition concerning $\derive{m}$ and $\prsreply{m}{\pi}$ is that on 
a request to service $s$ to process method $m$:
\begin{itemize}
\item
if $\prsreply{m}{\pi}(s) = \True$, $s$ processes $m$, produces the reply
$\True$ with probability $\pi$ and the reply $\False$ with probability
$1 - \pi$, and then proceeds as $\derive{m}(s)$;
\item
if $\prsreply{m}{\pi}(s) = \False$ for each $\pi \in \Prob$, $s$ is not 
able to process method $m$ and proceeds as $\emptyserv$.
\end{itemize}
The empty service $\emptyserv$ itself is unable to process any method.
A service is fully deterministic if, for all $m$, for all $s$, 
$\prsreply{m}{\pi}(s) = \True$ only if $\pi \in \set{0,1}$.

The assumptions with respect to services made above are the ones made 
before for the non-probabilistic case in e.g.~\cite{BM12b} adapted to 
the probabilistic case.

It is also assumed that a fixed but arbitrary set $\Foci$ of
\emph{foci} has been given.
Foci play the role of names of services in a service family. 

\SFA\ has the sorts, constants and operators from $\Sig{\ServAlg}$ and
in addition the sort $\ServFam$ of \emph{service families} and the 
following constant and operators:
\begin{itemize}
\item
the
\emph{empty service family} constant $\const{\emptysf}{\ServFam}$;
\item
for each $f \in \Foci$, the unary
\emph{singleton service family} operator
$\funct{\mathop{f{.}} \ph}{\Serv}{\ServFam}$;
\item
the binary
\emph{service family composition} operator
$\funct{\ph \sfcomp \ph}{\ServFam \x \ServFam}{\ServFam}$;
\item
for each $F \subseteq \Foci$, the unary
\emph{encapsulation} operator $\funct{\encap{F}}{\ServFam}{\ServFam}$.
\end{itemize}
We assume that there are infinitely many variables of sort $\Serv$,
including $s$, and infinitely many variables of sort $\ServFam$,
including $u,v,w$.
Terms are built as usual in the many-sorted case
(see e.g.~\cite{ST99a,Wir90a}).
We use prefix notation for the singleton service family operators and
infix notation for the service family composition operator.

The service family denoted by $\emptysf$ is the empty service family.
The service family denoted by a closed term of the form $f.t$ consists 
of one named service only, the service concerned is the service denoted 
by $t$, and the name of this service is $f$.
The service family denoted by a closed term of the form
$t \sfcomp t'$ consists of all named services that belong to either the
service family denoted by $t$ or the service family denoted by $t'$.
In the case where a named service from the service family denoted by
$t$ and a named service from the service family denoted by $t'$ have
the same name, they collapse to an empty service with the name
concerned.
The service family denoted by a closed term of the form $\encap{F}(t)$ 
consists of all named services with a name not in $F$ that belong to
the service family denoted by $t$.

The axioms of \SFA\ are given in 
Table~\ref{axioms-SFA}.%
\begin{table}[!t]
\caption{Axioms of \SFA}
\label{axioms-SFA}
{
\begin{eqntbl}
\begin{axcol}
u \sfcomp \emptysf = u                                 & \axiom{SFC1} \\
u \sfcomp v = v \sfcomp u                              & \axiom{SFC2} \\
(u \sfcomp v) \sfcomp w = u \sfcomp (v \sfcomp w)      & \axiom{SFC3} \\
f.s \sfcomp f.s' = f.\emptyserv                        & \axiom{SFC4}
\end{axcol}
\qquad
\begin{saxcol}
\encap{F}(\emptysf) = \emptysf                       & & \axiom{SFE1} \\
\encap{F}(f.s) = \emptysf            & \mif f \in F    & \axiom{SFE2} \\
\encap{F}(f.s) = f.s                 & \mif f \notin F & \axiom{SFE3} \\
\multicolumn{2}{@{}l@{\;\;}}
 {\encap{F}(u \sfcomp v) =
  \encap{F}(u) \sfcomp \encap{F}(v)}                   & \axiom{SFE4}
\end{saxcol}
\end{eqntbl}
}
\end{table}
In this table, $f$ stands for an arbitrary focus from $\Foci$ and
$F$ stands for an arbitrary subset of $\Foci$.
These axioms simply formalize the informal explanation given
above.

The constants, operators, and axioms of \SFA\ were presented for the 
first time in~\cite{BM09k}.

For the set $\BAct$ of basic actions, we now take 
$\set{f.m \where f \in \Foci, m \in \Meth}$.
Performing a basic action $f.m$ is taken as making a request to the
service named $f$ to process method $m$.

We combine \prBTA\ with \SFA\ and extend the combination with the 
following operators:
\begin{itemize}
\item
the binary \emph{use} operator
$\funct{\ph \sfuse \ph}{\Thr \x \ServFam}{\Thr}$;
\item
the unary \emph{abstraction} operator 
$\funct{\abstr{\Tau}}{\Thr}{\Thr}$;
\end{itemize}
and the axioms given in Tables~\ref{axioms-use} and~\ref{axioms-abstr},%
\begin{table}[!t]
\caption{Axioms for the use operator}
\label{axioms-use}
\begin{eqntbl}
\begin{axcol}
\DeadEnd \sfuse u = \DeadEnd                          & \axiom{prU1} \\
\Stop  \sfuse u = \Stop                               & \axiom{prU2} \\
(\Tau \bapf x) \sfuse u = \Tau \bapf (x \sfuse u)     & \axiom{prU3} \\
(\pcc{x}{f.m}{y}) \sfuse \encap{\set{f}}\hsp{-.1}(u) =
\pcc{(x \sfuse \encap{\set{f}}\hsp{-.1}(u))}
 {f.m}{(y \sfuse \encap{\set{f}}\hsp{-.1}(u))}        & \axiom{prU4} \\
(\pcc{x}{f.m}{y}) \sfuse (f.t \sfcomp \encap{\set{f}}\hsp{-.1}(u)) =
\Tau \bapf 
((\prc{x}{\pi}{y}) \sfuse 
 (f.\derive{m}t \sfcomp \encap{\set{f}}\hsp{-.1}(u))) \\ 
\hfill \mif \prsreply{m}{\pi}(t) = \True 
\phantom{\LAND{\pi \in \Prob}{}\,}                    & \axiom{prU5} \\
(\pcc{x}{f.m}{y}) \sfuse (f.t \sfcomp \encap{\set{f}}\hsp{-.1}(u)) = 
\Tau \bapf \DeadEnd
\hfill \mif \LAND{\pi \in \Prob}{} \prsreply{m}{\pi}(t) = \False 
                                                      & \axiom{prU6} \\
(\prc{x}{\pi}{y}) \sfuse u =
\prc{(x \sfuse u)}{\pi}{(y \sfuse u)}                 & \axiom{prU7} 
\end{axcol}
\end{eqntbl}
\end{table}
\begin{table}[!t]
\caption{Axioms for the abstraction operator}
\label{axioms-abstr}
\begin{eqntbl}
\begin{axcol}
\abstr{\Tau}(\Stop) = \Stop                             & \axiom{TA1} \\
\abstr{\Tau}(\DeadEnd) = \DeadEnd                       & \axiom{TA2} \\
\abstr{\Tau}(\Tau \bapf x) = \abstr{\Tau}(x)            & \axiom{TA3} \\
\abstr{\Tau}(\pcc{x}{f.m}{y}) =
\pcc{\abstr{\Tau}(x)}{f.m}{\abstr{\Tau}(y)}             & \axiom{TA4} \\
\abstr{\Tau}(\prc{x}{\pi}{y}) = 
\prc{\abstr{\Tau}(x)}{\pi}{\abstr{\Tau}(y)}             & \axiom{TA5}
\end{axcol}
\end{eqntbl}
\end{table}
and call the resulting theory \prTSI.
In these tables, $f$ stands for an arbitrary focus from $\Foci$, $m$ 
stands for an arbitrary method from $\Meth$, $\pi$ stands for an 
arbitrary probability from $\Prob$, and $t$ stands for an arbitrary term 
of sort $\Serv$.
The axioms formalize the informal explanation given below.
We use infix notation for the use operator.

The thread denoted by a closed term of the form $t \sfuse t'$ is the 
thread that results from processing the method of each basic action 
performed by the thread denoted by $t$ by the service with the focus of 
the basic action as its name in the service family denoted by $t'$  
each time that a service with the name in question really exists and as
long as the method concerned can be processed.
In the case that a service with the name in question does not really 
exist, the processing of a method is simply skipped (axiom~prU4).
When the method of a basic action performed by the thread can be 
processed by the named service, that service changes in accordance with 
the method and the thread is affected as follows: the basic action is 
turned into the internal action $\Tau$ and then an internal choice is 
made between the two ways to proceed according to the probabilities of 
the two possible reply values in the case of the method concerned 
(axiom prU5).
When the method of a basic action performed by the thread cannot be
processed by the named service, inaction occurs after the basic action 
is turned into the internal action $\Tau$ (axiom prU6).

The thread denoted by a closed term of the form $\abstr{\Tau}(t)$ is
the thread that results from concealing the presence of the internal 
action $\Tau$ in the thread denoted by $t$.

The use operator and the abstraction operator are adopted from the 
extension of \BTA\ with thread-service interaction presented before
in~\cite{BM12b}.
With the exception of axiom prU7, the axioms for the use operator are 
the ones given before for the non-probabilistic case in~\cite{BM12b} 
adapted to the probabilistic case.
With the exception of axiom TA5, the axioms for the abstraction operator 
are adopted from the extension of \BTA\ with thread-service interaction 
presented in~\cite{BM12b}.
Axiom prU7 and TA5 are new.

The following theorem concerns the question whether the operators added
to \prBTA\ in \prTSI\ are well axiomatized by the equations given in 
Tables~\ref{axioms-use} and~\ref{axioms-abstr} in the sense that these 
equations allow the projective limit model of \prBTA\ to be expanded to
a projective limit model of \prTSI.
\begin{theorem}
\label{theorem-use}
The operators added to \prBTA\ are well axiomatized, i.e.:
\begin{enumerate}
\item[(a)]
for all closed \prTSI\ terms $t$ of sort $\Thr$, there exists a closed 
\prBTA\ term $t'$ such that $t = t'$ is derivable from the axioms of 
\prTSI;
\item[(b)]
for all closed \prBTA\ terms $t$ and $t'$,
$t = t'$ is derivable from the axioms of \prBTA\ iff
$t = t'$ is derivable from the axioms of \prTSI;
\item[(c)]
for all closed \prTSI\ terms $t$ of sort $\Thr$, closed \prTSI\ terms 
$t'$ of sort $\ServFam$ and $n \in \Nat$, 
$\proj{n}(t \sfuse t') = \proj{n}(\proj{n}(t) \sfuse t')$ is derivable 
from the axioms of \prTSI\ and the following axioms for the unary 
operators $\proj{n}$ (which are explained below):%
\footnote
{Holding on to the usual conventions leads to the double use of the 
symbol $\pi$: without subscript it stands for a probability value 
and with subscript it stands for a projection operator.}
\begin{ldispl}
\begin{geqns}
\proj{0}(x) = \DeadEnd\;,                                             \\
\proj{n+1}(\DeadEnd) = \DeadEnd\;,                                    \\
\proj{n+1}(\Stop) = \Stop\;,                                          \\
\end{geqns}
\quad\;\;
\begin{geqns}
{} \\
\proj{n+1}(\pcc{x}{a}{y}) = \pcc{\proj{n}(x)}{a}{\proj{n}(y)}\;,      \\
\proj{n+1}(\prc{x}{\pi}{y}) = \prc{\proj{n+1}(x)}{\pi}{\proj{n+1}(y)}\;. 
\end{geqns}
\end{ldispl}%
where $n$ stands for an arbitrary natural number from $\Nat$, $a$ stands
for an arbitrary basic action from $\BActTau$, and $\pi$ is an arbitrary
probability from $\Prob$;
\item[(d)]
for all closed \prTSI\ terms $t$ of sort $\Thr$ and $n \in \Nat$, there 
exists a $k \in \Nat$ such that, for all $m \in \Nat$ with $m \geq k$, 
$\proj{n}(\abstr{\Tau}(t)) = \proj{n}(\abstr{\Tau}(\proj{m}(t)))$ 
is derivable from the axioms of \prTSI\ and the axioms for the operators 
$\proj{n}$ introduced in part~(c).
\end{enumerate}
\end{theorem}
\begin{proof}
Part~(a) is easily proved by induction on the structure of $t$, and in 
the case where $t$ is of the form $t_1 \sfuse t_2$ and the case where 
$t$ is of the form $\abstr{\Tau}(t_1)$ by induction on the structure of 
$t_1$.
In the subcase where $t$ is of the form 
\mbox{$\pcc{t'_1}{a}{t'_1} \sfuse t_2$}, we need the easy to prove fact 
that, for each $f \in \Foci$ and closed term $t$ of sort $\ServFam$, 
either $t = \encap{f}(t)$ is derivable or there exists a closed term 
$t'$ of sort $\Serv$ such that $t = f.t' \sfcomp \encap{f}(t)$ is 
derivable.

In the case of part~(b), the implication from left to right follows 
immediately from the fact that the axioms of \prBTA\ are included in the 
axioms of \prTSI.
The implication from right to left is not difficult to see either. 
From the axioms of \prTSI\ that are not axioms of \prBTA, only axioms 
prU1, prU2, prU6, TA1, and TA2 may be applicable to a closed \prBTA\ 
term $t$.
If one of them is applicable, then the application yields an equation 
$t = t'$ in which $t'$ is not a closed \prBTA\ term.
Moreover, only the axiom whose application yielded $t = t'$ is 
applicable to $t'$, but now in the opposite direction.
Hence, applications of axioms of \prTSI\ that are not axioms of \prBTA\
do not yield additional equations.

By part~(a), it is sufficient to prove parts~(c) and~(d) for all closed 
\prBTA\ terms $t$.
Parts~(c) and~(d) are easily proved by induction on the structure of 
$t$, and in each case by case distinction between $n = 0$ and 
$n > 0$.
In the proof of both parts, we repeatedly need the easy to prove fact 
that, for all closed \prBTA\ terms $t$ and $n \in \Nat$, 
$\proj{n}(t) = \proj{n}(\proj{n}(t))$ is derivable.
In the proof of part~(c), in the case where $t$ is of the form 
$\pcc{t_1}{a}{t_2}$, we need again the fact mentioned at the end of the
proof outline of part~(a).
\qed
\end{proof}
The unary operators $\proj{n}$ are called \emph{projection} operators.
The thread denoted by a closed term of the form $\proj{n}(t)$ is the 
thread that differs from the thread denoted by $t$ in that it becomes 
inactive as soon as it has performed $n$ actions.

By parts~(a) and~(b) of Theorem~\ref{theorem-use}, we know that the 
carrier of the projective limit model of \prBTA\ can serve as the 
carrier of a projective limit model of \prTSI\ if it is possible to 
define on this carrier operations corresponding to the added operators 
such that the added equations are  satisfied.
By  parts~(c) and~(d) of Theorem~\ref{theorem-use}, we know that it is 
possible to do so.
Thus, we know that the projective limit model of \prBTA\ can be expanded 
to a projective limit model of \prTSI.

The actual expansion goes along the same lines as in the 
non-probabilistic case (see~\cite{BM12b}).
An outline of this expansion is given in Appendix~\ref{app-prTAtsi}.
Because the depth of the approximations of a thread may decrease by 
abstraction, we do not have that, for all $n$ and $t$,
$\proj{n}(\abstr{\Tau}(t)) = \proj{n}(\abstr{\Tau}(\proj{n}(t)))$
is derivable.
However, it is sufficient that there exists a $k \in \Nat$ such that, 
for all $m \in \Nat$ with $m \geq k$, 
$\proj{n}(\abstr{\Tau}(t)) = \proj{n}(\abstr{\Tau}(\proj{m}(t)))$ 
is derivable (see also~\cite{BM12b}).

\section{A Probabilistic Program Notation}
\label{sect-prPGLB}

In this section, we introduce the probabilistic program notation 
\prPGLB\ (probabilistic PGLB).
In~\cite{BL02a}, a hierarchy of program notations rooted in program
algebra is presented.
One of the program notations that belong to this hierarchy is \PGLB\
(ProGramming Language B).
This program notation is close to existing assembly languages and has
relative jump instructions.
The program notation \prPGLB\ is \PGLB\ extended with probabilistic 
instructions that allow probabilistic choices to be made during the 
execution of instruction sequences. 

In \prPGLB, it is assumed that a fixed but arbitrary non-empty finite
set $\BInstr$ of \emph{basic instructions} has been given.
The intuition is that the execution of a basic instruction in most
instances modifies a state and in all instances produces a reply at its
completion.
The possible replies are the values $\True$ and $\False$, and the actual 
reply is in most instances state-dependent.
Therefore, successive executions of the same basic instruction may
produce different replies.
The set $\BInstr$ is the basis for the set of all instructions that may
appear in the instruction sequences considered in \prPGLB.
These instructions are called primitive instructions.

The program notation \prPGLB\ has the following primitive instructions:
\begin{itemize}
\item
for each $a \in \BInstr$, a \emph{plain basic instruction} $a$;
\item
for each $a \in \BInstr$, a \emph{positive test instruction} $\ptst{a}$;
\item
for each $a \in \BInstr$, a \emph{negative test instruction} $\ntst{a}$;
\item
for each $\pi \in \Prob$, a \emph{plain random choice instruction} 
$\prbsc{\pi}$;
\item
for each $\pi \in \Prob$, a \emph{positive random choice instruction} 
$\ptst{\prbsc{\pi}}$;
\item
for each $\pi \in \Prob$, a \emph{negative random choice instruction} 
$\ntst{\prbsc{\pi}}$;
\item
for each $l \in \Nat$, a \emph{forward jump instruction}
$\fjmp{l}$;
\item
for each $l \in \Nat$, a \emph{backward jump instruction}
$\bjmp{l}$;
\item
a \emph{termination instruction} $\halt$.
\end{itemize}
A \prPGLB\ instruction sequence has the form
$u_1 \conc \ldots \conc u_k$, where $u_1,\ldots,u_k$ are primitive
instructions of \prPGLB.

On execution of a \prPGLB\ instruction sequence, these primitive
instructions have the following effects:
\begin{itemize}
\item
the effect of a positive test instruction $\ptst{a}$ is that basic
instruction $a$ is executed and execution proceeds with the next
primitive instruction if $\True$ is produced and otherwise the next
primitive instruction is skipped and execution proceeds with the
primitive instruction following the skipped one -- if there is no
primitive instruction to proceed with, execution becomes inactive;
\item
the effect of a negative test instruction $\ntst{a}$ is the same as the
effect of $\ptst{a}$, but with the role of the value produced reversed;
\item
the effect of a plain basic instruction $a$ is the same as the effect of
$\ptst{a}$, but execution always proceeds as if $\True$ is produced;
\item
the effect of a positive random choice instruction $\ptst{\prbsc{\pi}}$ 
is that first $\True$ is produced with probability $\pi$ and $\False$ is 
produced with probability $1 - \pi$ and then execution proceeds with the 
next primitive instruction if $\True$ is produced and otherwise the next
primitive instruction is skipped and execution proceeds with the
primitive instruction following the skipped one -- if there is no
primitive instruction to proceed with, execution becomes inactive;
\item
the effect of a negative random choice instruction $\ntst{\prbsc{\pi}}$ 
is the same as the effect of $\ptst{\prbsc{\pi}}$, but with the role of 
the value produced reversed;
\item
the effect of a plain random choice instruction $\prbsc{\pi}$ is the 
same as the effect of $\ptst{\prbsc{\pi}}$, but execution always 
proceeds as if $\True$ is produced;
\item
the effect of a forward jump instruction $\fjmp{l}$ is that execution
proceeds with the $l^\mathrm{th}$ next primitive instruction -- if $l$
equals $0$ or there is no primitive instruction to proceed with,
execution becomes inactive;
\item
the effect of a backward jump instruction $\bjmp{l}$ is that execution
proceeds with the $l^\mathrm{th}$ previous primitive instruction -- if
$l$ equals $0$ or there is no primitive instruction to proceed with,
execution becomes inactive;
\item
the effect of the termination instruction $\halt$ is that execution 
terminates.
\end{itemize}

With the exception of the random choice instructions, the primitive
instructions of \prPGLB\ are adopted from \PGLB.
Counterparts of the random choice instructions are especially found in
probabilistic extensions of Dijkstra's guarded command language 
(see e.g.~\cite{HSM97a}).

In order to describe the behaviours produced by \prPGLB\ instruction 
sequences on execution, we need a service that behaves as a random 
Boolean generator.
This service is able to process the following methods:
\begin{itemize}
\item
for each $\pi \in \Prob$, a \emph{get random Boolean method} 
$\get(\pi)$.
\end{itemize}
For each $\pi \in \Prob$, the method $\get(\pi)$ can be explained as 
follows: the service produces the reply $\True$ with probability $\pi$ 
and the reply $\False$ with probability $1 - \pi$.

For the carrier of sort $\Serv$, we take the set 
$\set{\Random,\emptyserv}$.
For each $m \in \Meth$ and $\pi \in \Prob$, we take the functions 
$\derive{m}$ and $\prsreply{m}{\pi}$ such that:
\begin{ldispl}
\begin{geqns}
\derive{\get(\pi)}(\Random)        = \Random\;, \\
\prsreply{\get(\pi)}{\pi}(\Random) = \True\;, 
\end{geqns}
\qquad
\begin{gceqns}
\derive{m}(\Random) = \emptyserv & \mathrm{if} \; 
m \not\in \set{\get(\pi) \where \pi \in \Prob}\;, \\
\prsreply{m}{\pi}(\Random) = \False & \mathrm{if} \; 
m \neq \get(\pi)\;.
\end{gceqns}
\end{ldispl}%
Moreover, we take the name $\Random$ used above to denote the element of
the carrier of sort $\Serv$ that differs from $\emptyserv$ for a 
constant of sort $\Serv$.
It is assumed that $\get(\pi) \in \Meth$ for each $\pi \in \Prob$.
It is also assumed that $\random \in \Foci$.

The behaviours produced by \prPGLB\ instruction sequences on execution
are considered to be probabilistic threads, with the basic instructions 
taken as basic actions.
The \emph{thread extraction} operation $\extr{\ph}$ defines, for each
\prPGLB\ instruction sequence, the behaviour produced on its execution.
The thread extraction operation is defined by
\begin{ldispl}
\extr{u_1 \conc \ldots \conc u_k} =
\abstr{\Tau}(\extr{1,u_1 \conc \ldots \conc u_k} \sfuse \random.\Random)\;,
\end{ldispl}%
where $\extr{\ph,\ph}$ is defined by the equations given in
Table~\ref{axioms-thread-extr} (for $a \in \BInstr$, $\pi \in \Prob$, and 
$l,i \in \Nat$)%
\begin{table}[!t]
\caption{Defining equations for the thread extraction operation}
\label{axioms-thread-extr}
\begin{eqntbl}
\begin{aceqns}
\extr{i,u_1 \conc \ldots \conc u_k} & = & \DeadEnd
& \mif \Lnot 1 \leq i \leq k \\
\extr{i,u_1 \conc \ldots \conc u_k} & = &
a \bapf \extr{i+1,u_1 \conc \ldots \conc u_k}
& \mif u_i = a \\
\extr{i,u_1 \conc \ldots \conc u_k} & = &
\pcc{\extr{i+1,u_1 \conc \ldots \conc u_k}}{a}
    {\extr{i+2,u_1 \conc \ldots \conc u_k}}
& \mif u_i = \ptst{a} \\
\extr{i,u_1 \conc \ldots \conc u_k} & = &
\pcc{\extr{i+2,u_1 \conc \ldots \conc u_k}}{a}
    {\extr{i+1,u_1 \conc \ldots \conc u_k}}
& \mif u_i = \ntst{a} \\
\extr{i,u_1 \conc \ldots \conc u_k} & = &
\random.\get(\pi) \bapf \extr{i+1,u_1 \conc \ldots \conc u_k}
& \mif u_i = \prbsc{\pi} \\
\extr{i,u_1 \conc \ldots \conc u_k} & = &
\pcc{\extr{i+1,u_1 \conc \ldots \conc u_k}}{\random.\get(\pi)}
    {\extr{i+2,u_1 \conc \ldots \conc u_k}}
& \mif u_i = \ptst{\prbsc{\pi}} \\
\extr{i,u_1 \conc \ldots \conc u_k} & = &
\pcc{\extr{i+2,u_1 \conc \ldots \conc u_k}}{\random.\get(\pi)}
    {\extr{i+1,u_1 \conc \ldots \conc u_k}}
& \mif u_i = \ntst{\prbsc{\pi}} \\
\extr{i,u_1 \conc \ldots \conc u_k} & = &
\extr{i+l,u_1 \conc \ldots \conc u_k}
& \mif u_i = \fjmp{l} \\
\extr{i,u_1 \conc \ldots \conc u_k} & = &
\extr{i \monus l,u_1 \conc \ldots \conc u_k}
& \mif u_i = \bjmp{l} \\
\extr{i,u_1 \conc \ldots \conc u_k} & = & \Stop
& \mif u_i = \halt 
\end{aceqns}
\end{eqntbl}
\end{table}%
\footnote
{We write $i \monus j$ for the monus of $i$ and $j$, i.e.\
 $i \monus j = i - j$ if $i \geq j$ and $i \monus j = 0$ otherwise.
}
and the rule that $\extr{i,u_1 \conc \ldots \conc u_k} = \DeadEnd$ if
$u_i$ is the beginning of an infinite jump chain.%
\footnote
{This rule can be formalized, cf.~\cite{BM07g}.}

If $1 \leq i \leq k$, 
$\abstr{\Tau}
  (\extr{i,u_1 \conc \ldots \conc u_k} \sfuse \random.\Random)$ can be 
read as the behaviour produced by $u_1 \conc \ldots \conc u_k$ on 
execution if execution starts at the $i^\mathrm{th}$ primitive 
instruction.
By default, execution starts at the first primitive instruction.

In~\cite{BM09f}, we proposed several kinds of probabilistic jump
instructions (bounded and unbounded, according to uniform probability 
distributions and geometric probability distributions).
The meaning of instruction sequences from extensions of \prPGLB\ with
these kinds of probabilistic instructions can be given by a translation 
to instruction sequences from \prPGLB.

\section{Probabilistic Strategic Interleaving of Threads}
\label{sect-strategic-interleaving}

Multi-threading refers to the concurrent existence of several threads
in a program under execution.
It is the dominant form of concurrency provided by contemporary 
programming languages such as Java~\cite{GJSB00a} and C\#~\cite{HWG03a}.
Theories of concurrent processes such as \ACP~\cite{BW90}, 
CCS~\cite{Mil89}, and CSP~\cite{Hoa85} are based on arbitrary 
interleaving.
In the case of multi-threading, more often than not some interleaving 
strategy is used.
We abandon the point of view that arbitrary interleaving is the most 
appropriate abstraction when dealing with multi-threading.
The following points illustrate why we find difficulty in taking that
point of view:
(a)~whether the interleaving of certain threads leads to inactiveness
depends on the interleaving strategy used;
(b)~sometimes inactiveness occurs with a particular interleaving
strategy whereas arbitrary interleaving would not lead to inactiveness, 
and vice versa.
Demonstrations of (a) and (b) are given in~\cite{BM04c}
and~\cite{BM05c}, respectively.

The probabilistic features of \prBTA\ allow it to be extended with 
interleaving strategies that correspond to probabilistic scheduling 
algorithms.
In this section, we take up the extension of \prBTA\ with such
probabilistic interleaving strategies.
The presented extension covers an arbitrary probabilistic interleaving 
strategy that can be represented in the way that is explained below.

We write $\BActTau'$ for $\BActTau \union \set{\nt,\Stop,\DeadEnd}$ and 
we write $\Hist$ for $\seqof{(\Natpos \x \Natpos)}$.%
\footnote
{We write $\Natpos$ for the set $\set{n \in \Nat \where n \geq 1}$ of 
positive natural numbers.}
The elements of $\Hist$ are called \emph{interleaving histories}.
The intuition concerning interleaving histories is as follows:
if the $j$th pair of an interleaving history is $\tup{i,n}$, then the 
$i$th thread got a turn in the $j$th interleaving step and after its
turn there were $n$ threads to be interleaved.

With regard to interleaving of threads, it is assumed that the following 
has been given:
\begin{itemize}
\sloppy
\item
a set $S$; 
\item
an indexed family of functions $\indfam{\sched{n}}{n \in \Natpos}$
where 
$\funct{\sched{n}}{\Hist \x S}{(\set{1,\ldots,n} \to \Prob)}$
for each $n \in \Natpos$;
\item
an indexed family of functions $\indfam{\updat{n}}{n \in \Natpos}$
where 
$\funct{\updat{n}}{\Hist \x S \x \set{1,\ldots,n} \x \BActTau'}{S}$
for each $n \in \Natpos$.
\end{itemize}
The elements of $S$ are called \emph{control states}, $\sched{n}$ is 
called an \emph{abstract scheduler} (for $n$ threads), and $\updat{n}$ 
is called a \emph{control state transformer} (for $n$ threads).
The intuition concerning $S$, $\indfam{\sched{n}}{n \in \Natpos}$, and 
$\indfam{\updat{n}}{n \in \Natpos}$ is as follows:
\begin{itemize}
\item
the control states from $S$ encode data relevant to the interleaving 
strategy (e.g., for each of the threads being interleaved, the set of 
all foci naming services on which it currently keeps a lock);
\item
for each $h \in \Hist$ and $s \in S$, $\sched{n}(h,s)$ is the 
probability distribution on $n$ threads that assigns to each of the 
threads the probability that it gets the next turn after history $h$ in 
state $s$;
\item
for each $h \in \Hist$, $s \in S$, $i \in \set{1,\ldots,n}$, and 
$a \in \BActTau'$, $\updat{n}(h,s,i,a)$ is the control state that arises 
after history $h$ in state $s$ on the $i$th thread doing~$a$.
\end{itemize}
Thus, $S$, $\indfam{\sched{n}}{n \in \Natpos}$, and 
$\indfam{\updat{n}}{n \in \Natpos}$ provide a way to represent a 
probabilistic interleaving strategy.
The abstraction of a scheduler used here is essentially the notion of
a scheduler defined in~\cite{SS00a}.

We extend \prBTA\ with the following operators:
\begin{itemize}
\item
the ternary \emph{forking postconditional composition} operator
$\funct{\pcc{\ph}{\newthr{\ph}}{\ph}}{\Thr \x \Thr \x \Thr}{\Thr}$; 
\item
for each $n \in \Natpos$, $h \in \Hist$, and $s \in S$,
the $n$-ary \emph{strategic interleaving} operator
$\funct{\siop{n}{h}{s}}{\Thr \x \cdots \x \Thr}{\Thr}$;
\item
for each $n,i \in \Natpos$ with $i \leq n$, $h \in \Hist$, and 
$s \in S$,
the $n$-ary \emph{positional strategic interleaving} operator
$\funct{\posmop{n}{i}{h}{s}}{\Thr \x \cdots \x \Thr}{\Thr}$;
\item
the unary \emph{deadlock at termination} operator
$\funct{\stdop}{\Thr}{\Thr}$;
\end{itemize}
and the axioms given in Table~\ref{axioms-strategic-interleaving},%
\footnote
{We write
 $\emptyseq$ for the empty sequence, 
 $d$ for the sequence having $d$ as sole element, and 
 $\alpha \concat \alpha'$ for the concatenation of sequences $\alpha$ 
 and $\alpha'$.  
 We assume that the usual identities, such as
 $\emptyseq \concat \alpha = \alpha$ and
 $(\alpha \concat \alpha') \concat \alpha'' =
  \alpha \concat (\alpha' \concat \alpha'')$, hold.
}%
\begin{table}[!t]
\caption{Axioms for strategic interleaving}
\label{axioms-strategic-interleaving}
\begin{eqntbl}
\begin{axcol}
\si{n}{h}{s}{x_1,\ldots,x_n} =
\PRC{i=1}{n}{\sched{n}(h,s)(i)}{\posm{n}{i}{h}{s}{x_1,\ldots,x_n}}
                                                     & \axiom{prSI1}  
\eqnsep
\posm{1}{i}{h}{s}{\DeadEnd} = \DeadEnd               & \axiom{prSI2} \\
\posm{n+1}{i}{h}{s}
 {x_1,\ldots,x_{i-1},\DeadEnd,x_{i+1},\ldots,x_{n+1}} =
\\ \quad 
\std{\si{n}{h \concat \tup{i,n}}{\updat{n+1}(h,s,i,\DeadEnd)}
      {x_1,\ldots,x_{i-1},x_{i+1},\ldots,x_{n+1}}}   & \axiom{prSI3} \\
\posm{1}{i}{h}{s}{\Stop} = \Stop                     & \axiom{prSI4} \\
\posm{n+1}{i}{h}{s}{x_1,\ldots,x_{i-1},\Stop,x_{i+1},\ldots,x_{n+1}} =
\\ \quad
\si{n}{h \concat \tup{i,n}}{\updat{n+1}(h,s,i,\Stop)}
 {x_1,\ldots,x_{i-1},x_{i+1},\ldots,x_{n+1}}         & \axiom{prSI5} \\
\posm{n}{i}{h}{s}
 {x_1,\ldots,x_{i-1},\pcc{x_i'}{\newthr{x}}{x_i''},x_{i+1},\ldots,x_n} =
\\ \quad
\Tau \bapf
\si{n+1}{h \concat \tup{i,n+1}}{\updat{n}(h,s,i,\nt)}
 {x_1,\ldots,x_{i-1},x_i',x_{i+1},\ldots,x_n,x}        
                                                     & \axiom{prSI6} \\
\posm{n}{i}{h}{s}
 {x_1,\ldots,x_{i-1},\pcc{x_i'}{a}{x_i''},x_{i+1},\ldots,x_n} =
\\ \quad
\pcc{\si{n}{h \concat \tup{i,n}}{\updat{n}(h,s,i,a)}
      {x_1,\ldots,x_{i-1},x_i',x_{i+1},\ldots,x_n}\\ \qquad}{a}         
    {\\ \quad \si{n}{h \concat \tup{i,n}}{\updat{n}(h,s,i,a)}
               {x_1,\ldots,x_{i-1},x_i'',x_{i+1},\ldots,x_n}}        
                                                     & \axiom{prSI7} \\
\posm{n}{i}{h}{s}
 {x_1,\ldots,x_{i-1},\prc{x_i'}{\pi}{x_i''},x_{i+1},\ldots,x_n} =
\\ \quad
\prc{\posm{n}{i}{h}{s}
      {x_1,\ldots,x_{i-1},x_i',x_{i+1},\ldots,x_n}\\ \qquad}{\pi}         
    {\\ \quad \posm{n}{i}{h}{s}
               {x_1,\ldots,x_{i-1},x_i'',x_{i+1},\ldots,x_n}}        
                                                     & \axiom{prSI8} 
\eqnsep
\std{\DeadEnd} = \DeadEnd                            & \axiom{DT1} \\
\std{\Stop} = \DeadEnd                               & \axiom{DT2} \\
\std{\pcc{x}{\newthr{z}}{y}} =
       \pcc{\std{x}}{\newthr{\std{z}}}{\std{y}}      & \axiom{DT3} \\
\std{\pcc{x}{a}{y}} = \pcc{\std{x}}{a}{\std{y}}      & \axiom{DT4} \\
\std{\prc{x}{\pi}{y}} = \prc{\std{x}}{\pi}{\std{y}}  & \axiom{DT5} 
\end{axcol}
\end{eqntbl}
\end{table}
and call the resulting theory \prTA.
In this table, $n$ and $i$ stand for arbitrary numbers from $\Natpos$ 
with $i \leq n$, $h$ stands for an arbitrary interleaving history from 
$\Hist$, $s$ stands for an arbitrary control state from $S$, $a$ stands 
for an arbitrary basic action from $\BActTau$, and $\pi$ stands for an 
arbitrary probability from $\Prob$.

The forking postconditional composition operator has the same shape as
the postconditional composition operators introduced in
Section~\ref{sect-prBTA}.
Formally, no basic action is involved in forking postconditional 
composition.
However, for an operational intuition, in $\pcc{t}{\newthr{t''}}{t'}$,
$\newthr{t''}$ can be considered a thread forking action.
It represents the act of forking off thread $t''$.
Like with real basic actions, a reply is produced upon performing a 
thread forking action.

The thread denoted by a closed term of the form 
$\si{n}{h}{s}{t_1,\ldots,t_n}$ is the thread that results from 
interleaving of the $n$ threads denoted by $t_1,\ldots,t_n$ after
history $h$ in state $s$, according to the interleaving strategy 
represented by $S$, $\indfam{\sched{n}}{n \in \Natpos}$, and 
$\indfam{\updat{n}}{n \in \Natpos}$.
By the interleaving, a number of threads is turned into a single thread.
In this single thread, the internal action $\Tau$ arises as a residue of 
each thread forking action encountered.
Moreover, the possibility that $\False$ is produced as a reply upon 
performing a thread forking action is ignored.
This reflects our focus on the case where capacity problems with respect 
to thread forking never~arise.

The positional strategic interleaving operators are auxiliary operators 
used to axiomatize the strategic interleaving operators.
The role of the positional strategic interleaving operators in the 
axiomatization is similar to the role of the left merge operator found 
in process algebra (see e.g.~\cite{BW90}).
The deadlock at termination operator is an auxiliary operator as well.
It is used in axiom prSI3 to express that in the event of inactiveness 
of one thread, the whole become inactive only after all other threads 
have terminated or become inactive.
The thread denoted by a closed term of the form $\std{t}$ is the thread
that results from turning termination into inactiveness in the thread
denoted by $t$.

The forking postconditional composition operator and the deadlock at 
termination operator are adopted from earlier extensions of \BTA\ with 
strategic interleaving.
The strategic interleaving operators and the positional strategic 
interleaving operators are not adopted from earlier extensions of \BTA\ 
with strategic interleaving.
To our knowledge, no probabilistic process algebras with counterparts of 
these operators has been proposed until now. 
Axioms prSI1--prSI8 and DT5 are new.
Axioms DT1--DT4 are adopted from the extension of \BTA\ with strategic
interleaving and thread forking presented in~\cite{BM06c}.

Consider the case where $S$ is a singleton set, 
for each $n \in \Natpos$, $\sched{n}$ is defined by
\begin{ldispl}
\begin{gceqns}
\sched{n}(\emptyseq,s)(i) = 1\;
 & \mathrm{if} \;i = 1\;,  \\
\sched{n}(\emptyseq,s)(i) = 0\;
 & \mathrm{if} \;i \neq 1\;, \\
\sched{n}(h \concat \tup{j,n},s)(i) = 1\;
 & \mathrm{if} \;i = (j + 1) \bmod n\;, \\
\sched{n}(h \concat \tup{i,n},s)(i) = 0\;
 & \mathrm{if} \;i \neq (j + 1) \bmod n
\end{gceqns}
\end{ldispl}%
and, $\updat{n}$ is defined by 
\begin{ldispl}
\updat{n}(h,s,i,a) = s\;.
\end{ldispl}%
In this case, the interleaving strategy corresponds to the 
round-robin scheduling algorithm.
This deterministic interleaving strategy is called cyclic interleaving 
in our earlier work on interleaving strategies (see e.g.~\cite{BM04c}).
In the current setting, an interleaving strategy is deterministic if, 
for all $n$, for all $h$, $s$, and $i$, 
$\sched{n}(h,s)(i) \in \set{0,1}$.
In the case that $S$ and $\updat{n}$ are as above, but $\sched{n}$ is 
defined by
\begin{ldispl}
\begin{gceqns}
\sched{n}(\emptyseq,s)(i) = 1\;
 & \mathrm{if} \;i = 1\;,  \\
\sched{n}(\emptyseq,s)(i) = 0\;
 & \mathrm{if} \;i \neq 1\;, \\
\sched{n}(h \concat \tup{j,n},s)(i) = 1 / n\;
 & \mathrm{if} \;i \leq n\;, \\
\sched{n}(h \concat \tup{i,n},s)(i) = 0\;
 & \mathrm{if} \;i > n\;,
\end{gceqns}
\end{ldispl}%
the interleaving strategy is a purely probabilistic one.
The probability distribution used is a uniform distribution.

More advanced strategies can be obtained if the scheduling makes use of 
the whole interleaving history and/or the control state.
For example, the individual lifetimes of the threads to be interleaved 
and their creation hierarchy can be taken into account by making use of 
the whole interleaving history.
Individual properties of the threads to be interleaved that depend on 
the actions performed by them can be taken into account by making use of 
the control state.
By doing so, interleaving strategies are obtained which, to a certain 
extent, can be affected by the threads to be interleaved.

Henceforth, we will write \prBTAnt\ for \prBTA\ extended with the 
forking postconditional composition operator.
The projective limit model of \prBTAnt\ is constructed like the 
projective limit model of \prBTA.
An outline of the projective limit model of \prBTAnt\ is given in 
Appendix~\ref{app-prTAsi}.

The following theorem concerns the question whether the operators added
to \prBTAnt\ are well axiomatized by the equations given in 
Table~\ref{axioms-strategic-interleaving} in the sense that these 
equations allow the projective limit model of \prBTAnt\ to be expanded 
to a projective limit model of \prTA.
\begin{theorem}
\label{theorem-si}
The operators added to \prBTAnt\ are well axiomatized, i.e.:
\begin{enumerate}
\item[(a)]
for all closed \prTA\ terms $t$, there exists a closed \prBTAnt\ term 
$t'$ such that $t = t'$ is derivable from the axioms of \prTA;
\item[(b)]
for all closed \prBTAnt\ terms $t$ and $t'$,
$t = t'$ is derivable from the axioms of \prBTAnt\ iff
$t = t'$ is derivable from the axioms of \prTA;
\item[(c)]
for all $m ,i\in \Natpos$ with $i \leq m$, $h \in \Hist$, $s \in S$, 
closed \prTA\ terms $t_1,\linebreak[2]\ldots,t_m$ and $n \in \Nat$, 
$\proj{n}(\siop{m}{h}{s}(t_1,\ldots,t_m)) =
 \proj{n}(\siop{m}{h}{s}(\proj{n}(t_1),\ldots,\proj{n}(t_m)))$ and
$\proj{n}(\posmop{m}{i}{h}{s}(t_1,\ldots,t_m)) =
 \proj{n}(\posmop{m}{i}{h}{s}(\proj{n}(t_1),\ldots,\proj{n}(t_m)))$ 
are derivable from the axioms of \prTA, the axioms for the operators 
$\proj{n}$ introduced in Theorem~\ref{theorem-use}, and
the following axiom:
\begin{ldispl}
\begin{geqns}
\proj{n+1}(\pcc{x}{\newthr{z}}{y}) = 
\pcc{\proj{n+1}(x)}{\newthr{\proj{n+1}(z)}}{\proj{n+1}(y)}\;,              
\end{geqns}
\end{ldispl}%
where $n$ stands for an arbitrary natural number from $\Nat$;
\item[(d)]
for all closed \prTA\ terms $t$ and $n \in \Nat$, 
$\proj{n}(\stdop(t)) = \proj{n}(\stdop(\proj{n}(t)))$ 
is derivable from the axioms of \prTA, the axioms for the operators 
$\proj{n}$ introduced in Theorem~\ref{theorem-use}, and the axiom
introduced in part~(c).
\end{enumerate}
\end{theorem}
\begin{proof}
Part~(a) is straightforwardly proved by induction on the structure of 
$t$, and then in the case where $t$ is of the form 
$\posm{n}{i}{h}{s}{t_1,\ldots,t_n}$ by induction on the sum of the 
lengths of $t_1,\ldots,t_n$ and case distinction on the structure of 
$t_i$ and in the case where $t$ is of the form $\std{t_1}$ by induction 
on the structure of $t_1$.
The proof of the case where $t$ is of the form
$\posm{n}{i}{h}{s}{t_1,\ldots,t_n}$ reveals that occurrences of the 
forking postconditional composition operator get eliminated if $t$ is of
that form.

In the case of part~(b), the implication from left to right follows 
immediately from the fact that the axioms of \prBTAnt\ are included in 
the axioms of \prTA.
The implication from right to left is not difficult to see either. 
From the axioms of \prTA\ that are not axioms of \prBTAnt, only axioms 
prSI2, prSI4, DT1, and DT2 may be applicable to a closed \prBTAnt\ term 
$t$.
If one of them is applicable, then the application yields an equation 
$t = t'$ in which $t'$ is not a closed \prBTAnt\ term.
Moreover, only the axiom whose application yielded $t = t'$ is 
applicable to $t'$, but now in the opposite direction.
Hence, applications of axioms of \prTA\ that are not axioms of \prBTAnt\
do not yield additional equations.

By part~(a), it is sufficient to prove part~(c) for all closed \prBTAnt\ 
terms $t_1,\ldots,t_m$.
The derivability of the second equation is straightforwardly proved by 
induction on the sum of the lengths of $t_1,\ldots,t_n$ and case 
distinction on the structure of $t_i$, and in each case by case 
distinction between $n = 0$ and $n > 0$.
The derivability of the first equation now follows immediately using the 
axioms of the operators $\proj{n}$.
In the proofs, we repeatedly need the easy to prove fact that, for all 
closed \prBTAnt\ terms $t$ and $n \in \Nat$, 
$\proj{n}(t) = \proj{n}(\proj{n}(t))$ is derivable.

By part~(a), it is sufficient to prove part~(d) for all closed \prBTAnt\ 
terms $t$.
Part~(d) is easily proved by induction on the structure of $t$, and in 
each case by case distinction between $n = 0$ and $n > 0$.
In the proof, we need again the fact mentioned at the end of the proof
outline of part~(c).
\qed
\end{proof}

By Theorem~\ref{theorem-si}, we know that the projective limit model 
of \prBTAnt\ can be expanded to a projective limit model of \prTA.
An outline of this expansion is given in Appendix~\ref{app-prTAsi}.

\section{Concluding Remarks}
\label{sect-concl}

We have added probabilistic features to \BTA\ and its extensions with 
thread-service interaction and strategic interleaving.
Thus, we have paved the way for rigorous investigation of issues related
to probabilistic computation thinking in terms of instruction sequences
and rigorous investigation of probabilistic interleaving strategies.
As an example of the use of \prTSI, the probabilistic version of the 
extension of \BTA\ with thread-service interaction, we have added the 
most basic kind of probabilistic instructions proposed in~\cite{BM09f} 
to a program notation rooted in \PGA\ and have given a formal definition 
of the behaviours produced by the instruction sequences from the 
resulting program notation under excution with the help of \prTSI.

We enumerate neither the numerous issues relating to probabilistic 
computation in areas such as computability and complexity of 
computational problems, efficiency of algorithms, and verification of 
programs that could be investigated thinking in terms of instruction 
sequences nor the numerous probabilistic scheduling algorithms that 
could be investigated in \prTA, the probabilistic generalization of the 
extensions of \BTA\ with strategic interleaving.

However, we mention interesting options for future work that are of a 
different kind:
(a)~clarifying analyses of relevant probabilistic algorithms, such as 
the Miller-Rabin probabilistic primality test~\cite{Rab76a}, using 
probabilistic instruction sequences or non-probabilistic instruction 
sequences and probabilistic services and
(b)~explanations of relevant quantum algorithms, such as Shor's integer
factorization algorithm~\cite{Sho94a}, by first giving a clarifying
analysis using probabilistic instruction sequences or non-probabilistic 
instruction sequences and probabilistic services and then showing how 
certain services involved in principle can be realized very efficiently 
with quantum computing.

Moreover, we believe that the development of program notations for 
probabilistic computation is a useful preparation for the development of 
program notations for quantum computation later on.
The development of program notations for quantum computation that have 
their origins in instruction sequences could constitute a valuable 
complement to other developments with respect to quantum computation, 
which for the greater part boil down to mere adaptation of earlier 
developments with respect to classical computation to the potentialities 
of quantum physics (see e.g.~\cite{Gay06a}).

In fact, \prBTA\ is a process algebra tailored to the behaviours 
produced by probabilistic instruction sequences under execution.
Because \prBTA\ offers probabilistic choices of the generative variety 
(see~\cite{GSS95a}) and no non-deterministic choices, it is most closely 
related to the probabilistic process algebra prBPA presented 
in~\cite{BBS95a}.
To our knowledge, thread-service interaction and strategic interleaving 
as found in \prTSI\ and \prTA\ are mechanisms for interaction and 
concurrency that are quite different from those found in any theory or 
model of processes.
This leaves almost nothing to be said about related work.

The very limited extent of related work is due to two conscious choices:
(a)~the limitation of the scope to behaviours produced by programs under 
execution and
(b)~the limitation of the scope to the form of interleaving concurrency 
that is relevant to the behaviours of multi-threaded programs under 
execution.
However, something unexpected remains to be mentioned as related work, 
to wit the work on security of multi-threaded programs presented 
in~\cite{SS00a}.
Probabilistic strategic interleaving as found in \prTA\ is strongly 
inspired by the scheduler-dependent semantics of a simple programming 
language with support for multi-threading that we found in that paper.

It is noteworthy to mention something about the interpretation of 
\prBTA, \prTSI, and \prTA\ in a probabilistic version of a general 
process algebra such as \ACP, CCS or CSP.
It is crucial that probabilistic choice of the generative variety, 
non-deterministic choice, asynchronous parallel composition, abstraction 
from internal actions, and recursion are covered by the process algebra
used for the purpose of interpretation.
General process algebras that cover all this are rare.
To our knowledge, pACP$_\tau$~\cite{AG09a} is the only one that has been 
elaborated in sufficient depth.
However, interpretation of \prBTA, \prTSI, and \prTA\ in pACP$_\tau$ 
seems impossible to us.
The presence of asynchronous parallel composition based on arbitrary 
interleaving in pACP$_\tau$ precludes the proper form of abstraction 
from internal actions for interpretation of \prBTA, \prTSI, and \prTA. 

\subsection*{Acknowledgements}

We thank two anonymous referees for carefully reading a preliminary
version of this paper and for suggesting improvements of the
presentation of the paper.

\appendix

\section{Projective Limit Models}
\label{appendix}

In this appendix, we outline the construction of projective limit models 
for \prBTA, \prTSI, and \prTA.
In these model, which covers finite and infinite threads, threads are
represented by infinite sequences of finite approximations.
Guarded recursive specifications have unique solutions in these models.
We denote the interpretations of constants and operators in the models 
by the constants and operators themselves.

\subsection{Projective Limit Model of prBTA}
\label{app-prBTA}

We will write $\IMod{\prBTA}$ for the initial model of \prBTA\ and 
$\IThr{\prBTA}$ for the carrier of $\IMod{\prBTA}$.
$\IThr{\prBTA}$ consists of the equivalence classes of closed \prBTA\ 
terms with respect to derivable equality.
In other words, modulo derivable equality, $\IThr{\prBTA}$ is the set of
all closed \prBTA\ terms.
Henceforth, we will identify closed \prBTA\ terms with their equivalence
class where elements of $\IThr{\prBTA}$ are concerned.

Each element of $\IThr{\prBTA}$ represents a finite thread,  i.e.\ a
thread with a finite upper bound to the number of actions that it can
perform.
Below, we will construct a model that covers infinite threads as well.
In preparation for that, we define for all $n$ a function that cuts off
threads from $\IThr{\prBTA}$ after $n$ actions have been performed.

\sloppy
For each $n \in \Nat$, we define the \emph{projection} function
$\funct{\proj{n}}{\IThr{\prBTA}}{\IThr{\prBTA}}$, inductively as 
follows:
\begin{ldispl}
\begin{geqns}
\proj{0}(t) = \DeadEnd\;, 
\\
\proj{n+1}(\Stop)  = \Stop\;, 
\\
\proj{n+1}(\DeadEnd) = \DeadEnd\;, 
\\
\end{geqns}
\quad\;\;
\begin{geqns}
{} \\
\proj{n+1}(\pcc{t}{a}{t'}) = \pcc{\proj{n}(t)}{a}{\proj{n}(t')}\;,
\\
\proj{n+1}(\prc{t}{\pi}{t'}) = \prc{\proj{n+1}(t)}{\pi}{\proj{n+1}(t')}\;.
\end{geqns}
\end{ldispl}%
For $t \in \IThr{\prBTA}$, $\proj{n}(t)$ is called the $n$th projection
of $t$.
It can be thought of as an approximation of $t$.
If $\proj{n}(t) \neq t$, then $\proj{n+1}(t)$ can be thought of as the
closest better approximation of $t$.
If $\proj{n}(t) = t$, then $\proj{n+1}(t) = t$ as well.
For all $n \in \Nat$, we will write $\IThrn{n}{\prBTA}$ for
$\set{\proj{n}(t) \where t \in \IThr{\prBTA}}$.
Obviously, the projection functions defined above satisfy the axioms for 
the projection operators introduced in Theorem~\ref{theorem-use}.

In the projective limit model, which covers both finite and infinite
threads, threads are represented by \emph{projective sequences}, i.e.\ 
infinite se\-quences $\projseq{t_n}{n}$ of elements of $\IThr{\prBTA}$ 
such that $t_n \in \IThrn{n}{\prBTA}$ and $t_n = \proj{n}(t_{n+1})$ for 
all $n \in \Nat$.
In other words, a projective sequence is a sequence of which successive
components are successive projections of the same thread.
The idea is that any infinite thread is fully characterized by the
infinite sequence of all its finite approximations.
We will write $\PThr{\prBTA}$ for the set of all projective sequences 
over $\IThr{\prBTA}$, i.e.\ the set
\begin{ldispl}
\set{\projseq{t_n}{n} \where
     \LAND{n \in \Nat}{}
      (t_n \in \IThrn{n}{\prBTA} \Land t_n = \proj{n}(t_{n+1}))}\;.
\end{ldispl}%

The \emph{projective limit model} $\PMod{\prBTA}$ of \prBTA\ consists of 
the following:
\begin{itemize}
\item
the set $\PThr{\prBTA}$, the carrier of the projective limit model;
\item
an element of $\PThr{\prBTA}$ for each constant of \prBTA;
\item
an operation on $\PThr{\prBTA}$ for each operator of \prBTA;
\end{itemize}
where those elements of $\PThr{\prBTA}$ and operations on 
$\PThr{\prBTA}$ are defined as follows:
\begin{ldispl}
\begin{aeqns}
\Stop  & = & \projseq{\proj{n}(\Stop)}{n}\;,
\\
\DeadEnd & = & \projseq{\proj{n}(\DeadEnd)}{n}\;,
\\
\pcc{\projseq{t_n}{n}}{a}{\projseq{t'_n}{n}} & = &
\projseq{\proj{n}(\pcc{t_n}{a}{t'_n})}{n}\;,
\\
\prc{\projseq{t_n}{n}}{\pi}{\projseq{t'_n}{n}} & = &
\projseq{\proj{n}(\prc{t_n}{\pi}{t'_n})}{n}\;.
\end{aeqns}
\end{ldispl}%

\sloppy
It is straightforward to check that the constants are elements of
$\PThr{\prBTA}$ and the operations always yield elements of 
$\PThr{\prBTA}$.
It follows immediately from the construction of the projective limit
model of \prBTA\ that the axiom of \prBTA\ forms a complete 
axiomatization of this model for equations between closed terms.

\subsection{Projective Limit Model of prTA\boldmath{$_\mathrm{tsi}$}}
\label{app-prTAtsi}

We will write $\IMod{\SFA}$ for the free \SFA-extension of $\ServAlg$ 
and $\IMod{\prTSI}$ for the free \prTSI-extension of $\ServAlg$.

From the fact that the signatures of $\PMod{\prBTA}$ and $\IMod{\SFA}$
are disjoint, it follows, by the amalgamation result about expansions
presented as Theorem~6.1.1 in~\cite{Hod93a} (adapted to the many-sorted
case), that there exists a model of \prBTA\ combined with \SFA\ such 
that the restriction to the signature of \prBTA\ is $\PMod{\prBTA}$ and 
the restriction to the signature of \SFA\ is $\IMod{\SFA}$.

Let $\PMod{\prBTA{+}\SFA}$ be the model of \prBTA\ combined with \SFA\ 
referred to above.
Then the \emph{projective limit model} $\PMod{\prTSI}$ of \prTSI\ is
$\PMod{\prBTA{+}\SFA}$ expanded with the operations defined by
\begin{ldispl}
\begin{aeqns}
\projseq{t_n}{n} \sfuse S & = &
\projseq{\proj{n}(t_n \sfuse S)}{n}\;,
\\
\abstr{\Tau}(\projseq{t_n}{n}) & = &
\projseq{\lim_{k \to \infty} \proj{n}(\abstr{\Tau}(t_k))}{n}
\end{aeqns}
\end{ldispl}%
as interpretations of the additional operators of \prTSI.
On the right-hand side of these equations, the symbols $\sfuse$ and 
$\abstr{\Tau}$ denote the interpretation of the operators
$\sfuse$ and $\abstr{\Tau}$ in $\IMod{\prTSI}$.
In the second equation, the limit is the limit with respect to the 
discrete topology on $\IThr{\prBTA}$.

It is straightforward to check that the operations with which 
$\PMod{\prBTA}$ is expanded always yield elements of $\PThr{\prBTA}$.
It follows immediately from the construction of $\PMod{\prTSI}$ and 
Theorem~\ref{theorem-use} that $\PMod{\prTSI}$ is really a projective 
limit model of \prTSI.

\subsection{Projective Limit Model of prTA\boldmath{$_\mathrm{si}$}}
\label{app-prTAsi}

We will write $\IMod{\prBTAnt}$ for the initial model of \prBTAnt\ and 
$\IThr{\prBTAnt}$ for the carrier of $\IMod{\prBTAnt}$.
Moreover, we will write $\IMod{\prTA}$ for the initial model of \prTA.

With the projection functions $\proj{n}$ extended from $\IThr{\prBTA}$
to $\IThr{\prBTAnt}$ such that 
\begin{ldispl}
\proj{n+1}(\pcc{t}{\newthr{t''}}{t'}) = 
\pcc{\proj{n+1}(t)}{\newthr{\proj{n+1}(t'')}}{\proj{n+1}(t')}\;, 
\end{ldispl}%
the projective limit model $\PMod{\prBTAnt}$ of \prBTAnt\ is 
constructed from $\IMod{\prBTAnt}$ like the projective limit model 
$\PMod{\prBTA}$ of \prBTA\ is constructed from $\IMod{\prBTA}$.
The interpretation of the additional operator is the operation on 
$\PThr{\prBTAnt}$ defined as follows:    
\begin{ldispl}
\begin{aeqns}
\projseq{{t_1}_n}{n} \pccop{\newthr{\projseq{{t_2}_n}{n}}}
 \projseq{{t_3}_n}{n} & = &
\projseq{\proj{n}({t_1}_n \pccop{\newthr{{t_2}_n}} {t_3}_n)}{n}\;.
\end{aeqns}
\end{ldispl}%

The \emph{projective limit model} $\PMod{\prTA}$ of \prTA\ is 
$\PMod{\prBTAnt}$ expanded with the operations defined by
\begin{ldispl}
\begin{aeqns}
\siop{n}{h}{s}(\projseq{{t_1}_n}{n},\ldots,\projseq{{t_m}_n}{n}) & = &
\projseq{\proj{n}(\siop{n}{h}{s}({t_1}_n,\ldots,{t_m}_n))}{n}\;,
\\
\posmop{n}{i}{h}{s}(\projseq{{t_1}_n}{n},\ldots,\projseq{{t_m}_n}{n})
 & = &
\projseq{\proj{n}(\posmop{n}{i}{h}{s}({t_1}_n,\ldots,{t_m}_n))}{n}\;,
\\
\stdop(\projseq{t_n}{n}) & = &
\projseq{\proj{n}(\stdop(t_n))}{n}
\end{aeqns}
\end{ldispl}%
as interpretations of the additional operators of \prTA.
On the right-hand side of these equations, the symbols 
$\pccop{\newthr{\ph}}$, $\siop{n}{h}{s}$, $\posmop{n}{i}{h}{s}$, and
$\stdop$ denote the interpretation of the operators 
$\pccop{\newthr{\ph}}$, $\siop{n}{h}{s}$, $\posmop{n}{i}{h}{s}$, and
$\stdop$ in $\IMod{\prTA}$.

It is straightforward to check that the operations with which 
$\PMod{\prBTAnt}$ is expanded always yield elements of 
$\PThr{\prBTAnt}$.
It follows immediately from the construction of $\PMod{\prTA}$ and 
Theorem~\ref{theorem-si} that $\PMod{\prTA}$ is really a projective 
limit model of \prTA.

\bibliographystyle{splncs03}
\bibliography{IS}

\end{document}